\documentclass[10pt,journal,draftclsnofoot,onecolumn]{IEEEtran}

\usepackage{amsthm}
\usepackage{amsmath}
\usepackage{cite}
\usepackage{amsfonts}
\usepackage{graphicx}
\usepackage{latexsym}
\usepackage{amssymb}
\usepackage{stmaryrd}
\usepackage{amscd}
\usepackage{multirow}
\usepackage[small]{caption}
\usepackage{slashbox}

\newcommand{\deff}{\mbox{$\stackrel{\rm def}{=}$}}

\newcommand{\sbinom}[2]{\left[ \begin{array}{c} #1 \\ #2 \end{array} \right] }

\newcommand{\field}[1]{\mathbb{#1}}

\newcommand{\cA}{{\cal A}}
\newcommand{\cB}{{\cal B}}
\newcommand{\cC}{{\cal C}}
\newcommand{\cD}{{\cal D}}

\newcommand{\cV}{{\cal V}}

\newcommand{\sP}{\field{P}}

\newcommand{\sG}{\field{G}}

\DeclareMathAlphabet{\mathbfsl}{OT1}{cmr}{bx}{it}
\newcommand{\uuu}{\kern-1pt\mathbfsl{u}\kern-0.5pt}
\newcommand{\vvv}{\kern-1pt\mathbfsl{v}\kern-0.5pt}

\newcommand{\myboxplus}{\kern1pt\mbox{\small$\boxplus$}}

\makeatletter \DeclareRobustCommand{\sbinom}{\genfrac[]\z@{}}
\makeatother
\newcommand{\G}[2]{\sbinom{{#1}\kern-1pt}{{#2}\kern-1pt}}
\newcommand{\Gq}[2]{\sbinom{{#1}\kern-0.25pt}{{#2}\kern-0.25pt}}

\newcommand{\Ps}{\smash{{\sP\kern-2.0pt}_q\kern-0.5pt(n)}}
\newcommand{\sPs}{\smash{{\sP\kern-1.5pt}_q(n)}}
\newcommand{\Ptwo}{\smash{{\sP\kern-2.0pt}_2\kern-0.5pt(n)}}
\newcommand{\Ptwom}{\smash{{\sP\kern-2.0pt}_2\kern-0.5pt(m)}}
\newcommand{\Ptwonm}{\smash{{\sP\kern-2.0pt}_2\kern-0.5pt(n+m)}}
\newcommand{\Ptwoa}{\smash{{\sP\kern-2.0pt}_2\kern-0.5pt(1)}}
\newcommand{\Ptwob}{\smash{{\sP\kern-2.0pt}_2\kern-0.5pt(2)}}
\newcommand{\Ptwoc}{\smash{{\sP\kern-2.0pt}_2\kern-0.5pt(3)}}
\newcommand{\Ptwod}{\smash{{\sP\kern-2.0pt}_2\kern-0.5pt(4)}}
\newcommand{\Ptwoe}{\smash{{\sP\kern-2.0pt}_2\kern-0.5pt(5)}}
\newcommand{\Ptwof}{\smash{{\sP\kern-2.0pt}_2\kern-0.5pt(6)}}
\newcommand{\Ptwokm}{\smash{{\sP\kern-2.0pt}_2\kern-0.5pt(2k-1)}}
\newcommand{\Pone}{\smash{{\sP\kern-2.5pt}_2\kern-0.5pt(n{-}1)}}

\newcommand{\Gr}{\smash{{\sG\kern-1.5pt}_q\kern-0.5pt(n,k)}}
\newcommand{\Gi}{\smash{{\sG\kern-1.5pt}_q\kern-0.5pt(n,i)}}
\newcommand{\Gj}{\smash{{\sG\kern-1.5pt}_q\kern-0.5pt(n,j)}}
\newcommand{\Grmk}{\smash{{\sG\kern-1.5pt}_q\kern-0.5pt(n,n-k)}}
\newcommand{\Grdk}{\smash{{\sG\kern-1.5pt}_q\kern-0.5pt(2k,k)}}
\newcommand{\Grekappa}{\smash{{\sG\kern-1.5pt}_q\kern-0.5pt(n,e+1-\kappa)}}
\newcommand{\Grtwoekappa}{\smash{{\sG\kern-1.5pt}_q\kern-0.5pt(n,2e+1-\kappa)}}
\newcommand{\Gremkappa}{\smash{{\sG\kern-1.5pt}_q\kern-0.5pt(n,e-\kappa)}}
\newcommand{\Gn}{\smash{{\sG\kern-1.5pt}_2\kern-0.5pt(n,n{-}1)}}
\newcommand{\Gnq}{\smash{{\sG\kern-1.5pt}_q\kern-0.5pt(n,n{-}1)}}
\newcommand{\Gone}{\smash{{\sG\kern-1.5pt}_2\kern-0.5pt(n,1)}}
\newcommand{\Gqone}{\smash{{\sG\kern-1.5pt}_q\kern-0.5pt(n,1)}}
\newcommand{\GTwo}{\smash{{\sG\kern-1.5pt}_2\kern-0.5pt(n,k)}}
\newcommand{\GTwonk}[2]{{\smash{{\sG\kern-1.5pt}_2\kern-0.5pt({#1},{#2})}}}
\newcommand{\Gnk}{\smash{{\sG\kern-1.5pt}_2\kern-0.5pt(n,n{-}k)}}
\newcommand{\Greone}{\smash{{\sG\kern-1.5pt}_q\kern-0.5pt(n,e{+}1)}}
\newcommand{\Gretwo}{\smash{{\sG\kern-1.5pt}_q\kern-0.5pt(n,e{+}2)}}

\newcommand{\be}[1]{\begin{equation}\label{#1}}
\newcommand{\ee}{\end{equation}}

\newcommand{\Cref}[1]{Co\-rol\-la\-ry\,\ref{#1}}

\newcommand\scalemath[2]{\scalebox{#1}{\mbox{\ensuremath{\displaystyle #2}}}}

\newtheorem{theorem}{Theorem}
\newtheorem{lemma}{Lemma}
\newtheorem{remark}{Remark}
\newtheorem{cor}{Corollary}

\begin{document}

%\title{\huge Perfect Permutations Codes with the Kendall's $\tau$-Metric}
\title{Bounds on the Size of Permutation Codes \\ with the Kendall $\tau$-Metric\vspace{-1.0ex}}

\author{Sarit Buzaglo and Tuvi Etzion, Fellow, IEEE,%

\thanks{This work was supported in part by the United States --- Israel
Binational Science~Foundation (BSF), Jerusalem, Israel, under Grant 2012016.
This work is part of S. Buzaglo PhD dissertation performed at the Technion--Israel Institute of Technology.
The material in this paper was presented in part in the 2014 IEEE International Symposium on Information Theory, Honolulu,
Hawaii, June-July 2014.}

\thanks{S.~Buzaglo is with the Center for Magnetic Recording Research, University of California, San Diego, La Jolla, CA 92093-0401 USA (e-mail: \texttt{sbuzaglo@ucsd.edu}). }
\thanks{
T.~Etzion is with the Computer Science Department, Technion--Israel
Institute of Technology, Haifa 32000, Israel
(e-mail: \texttt{etzion@cs.technion.ac.il}).
}

}

\maketitle
\begin{abstract}

The rank modulation scheme has been proposed for efficient writing
and storing data in non-volatile memory storage. Error-correction
in the rank modulation scheme is done by considering permutation codes.
In this paper we consider codes in the set of all permutations on $n$ elements,
$S_n$, using the Kendall $\tau$-metric. The main goal of this paper is
to derive new bounds on the size of such codes. For this purpose we also consider
perfect codes, diameter perfect codes, and the size of optimal anticodes
in the Kendall $\tau$-metric, structures which have their own considerable interest. We prove that there are no perfect
single-error-correcting codes in $S_n$, where $n>4$ is a prime or
$4\leq n\leq 10$. We present lower bounds on the size of optimal anticodes with odd diameter. As a consequence we obtain a new upper bound
on the size of codes in $S_n$ with even minimum Kendall $\tau$-distance.
We present larger single-error-correcting codes than the known ones
in $S_5$ and $S_7$.
\end{abstract}


%%%%%%%%%%%%%%%%%%%%%%%%%%%%%%%%%
%%%
%%%     INTRODUCTION
%%%
%%%%%%%%%%%%%%%%%%%%%%%%%%%%%%%%%%%%

\section{Introduction}
Flash memory is a non-volatile technology that is both electrically
programmable and electrically erasable.
It incorporates a set of cells
maintained at a set of levels of charge to encode information.
While raising the charge level of a cell is an easy operation,
reducing the charge level requires the erasure of the whole block to which the cell belongs.
For this reason charge is injected into the cell over several iterations.
Such programming is slow and can cause errors since cells may be injected with extra unwanted charge.
Other common errors in flash memory cells are due to charge leakage and reading disturbance that
may cause charge to move from one cell to its adjacent cells.
In order to overcome these problems, the novel framework of \emph{rank modulation codes} was introduced in  \cite{JMSB09}.
In this setup the information is carried by the relative ranking of the
cells’ charge levels and not by the absolute values of the charge levels.
This allows for more efficient programming of cells, and coding by the ranking of the cells' levels
is more robust to charge leakage than coding by their actual values.
In this model codes are subsets of~$S_n$, the set of all permutations on $n$ elements,
where each permutation corresponds to a ranking of $n$ cells' levels.
Permutation codes were mainly studied in this context using three metrics, the infinity metric, the Ulam metric,
and the Kendall $\tau$-metric. Codes in $S_n$ under the infinity metric were
considered in \cite{KLT10,ShTs10,TaSc10,TaSc12}.
Anticodes in $S_n$ under the infinity metric were considered in \cite{Klo11,ShTs11,TaSc11}.
Codes in $S_n$ under the Ulam metric were considered in~\cite{FSM12}.
Permutation codes with other metrics were considered in many papers.
A survey on metrics related to permutations is given in~\cite{DeHu98}.

In this paper we consider codes using the Kendall $\tau$-metric~\cite{KeGi90}.
Under the Kendall $\tau$-metric, codes in $S_n$
with minimum distance $d$ should correct up to
$\left\lfloor\frac{d-1}{2}\right\rfloor$ errors that are
caused by small charge leakage and read disturbance. For large charge leakage and read
disturbance the Ulam metric is used~\cite{FSM12}. Let $P(n,d)$ denote the size of the largest code in $S_n$ with minimum Kendall $\tau$-distance $d$.
A~comprehensive work on error-correcting codes in $S_n$
using the Kendall $\tau$-metric and bounds on $P(n,d)$ were
considered in \cite{JSB10}. In that paper
there is also a construction of single-error-correcting
codes using codes in the Lee metric.
This method was generalized in \cite{BM10} for the construction of
$t$-error-correcting codes that are of optimal size
up to a constant factor, where $t$ is fixed. More constructions of error-correcting codes were given in~\cite{MBZ13}.
Systematic single-error-correcting
codes in~$S_n$ of size $(n-2)!$ were constructed in~\cite{ZJB12,ZSJB13}. The constructed codes are of optimal size,
assuming that perfect single-error-correcting codes do not exist.
But, only the nonexistence of perfect single-error-correcting codes for $n=4$ was proved.
Systematic $t$-error-correcting codes were studied in \cite{BYEB14,ZJB12,ZSJB13}.
Linear programming and semi-definite programming on permutation codes with the Kendall $\tau$-metric were considered in~\cite{LiHa12}.
Unfortunately, no bounds better than the sphere packing bound were found by these methods.

The main goal of this paper is to provide new bounds on the size of permutation codes in the Kendall $\tau$-metric. As part of this goal we will prove the nonexistence of perfect single-error-correcting codes in $S_n$ if $n$ is a prime. Although this improves the related upper bound on $P(n,3)$ only by one, such a result is of interest for itself. This is one of the two main results of this paper. The second main result is a new upper bound on the size of permutation codes in the Kendall $\tau$-metric, where the minimum distance is even. This bound is obtained by introducing the notion of anticodes in the Kendall $\tau$-metric and proving a related code-anticode theorem. Finally, we present two codes with minimum distance 3 in $S_5$ and $S_7$, which are considerably larger than the previous known codes. These codes are of special interest since the rank modulation scheme is more likely to be applicable for small values of~$n$.

The rest of this work is organized as follows. In
Section~\ref{sec:basic} we define the basic concepts for
the Kendall $\tau$-metric and for perfect codes.
In Section~\ref{sec:nonexist} we prove the nonexistence of
a perfect single-error-correcting code in $S_n$, using the Kendall $\tau$-metric,
where $n>4$ is a prime or $4\leq n\leq 10$. This is the first known result in this
direction and it shows that the sphere packing upper bound can not be attained in these cases.
In Section~\ref{sec:diameter} we establish the Delsarte's code-anticode bound for the
Kendall $\tau$-metric and examine
diameter perfect codes in $S_n$ for this metric.
We find the sizes of optimal anticodes in $S_n$ with diameter 2 and diameter 3
and consider the size of optimal anticodes for larger diameters as well.
Trivial diameter perfect codes are considered in some of these cases.
We combine these results with the code-anticode bound to
improve the known upper bound on the size of a code in~$S_n$
for even minimum distances.
In Section~\ref{sec:cyclic} we consider lower bounds on the size of permutation codes in the Kendall $\tau$-metric for small values of $n$. We search for such codes by forcing a structure and a certain automorphism group on the codes. Two large single-error-correcting codes for $n=5$ and $n=7$ are constructed in this way and yield an improvement on the related lower bounds.
We conclude in Section~\ref{sec:conclusion}, where we also present some questions for future research.

%%%%%%%%%%%%%%%%%%%%%%%%%%%%%%
%%
%%      BASIC CONCEPTS
%%
%%%%%%%%%%%%%%%%%%%%%%%%%%%%%%%%%%

\section{Basic Concepts}
\label{sec:basic}

Let $S_n$ be the set of all permutations on the set of $n$ elements $[n]\deff\{1,2,\ldots,n\}$.
We denote a permutation $\sigma\in S_n$ by $\sigma=[\sigma(1),\sigma(2),\ldots,\sigma(n)]$.
For two permutations $\sigma , \pi \in S_n$, their multiplication $\pi\circ \sigma$ is defined as the
composition of $\sigma$ on $\pi$, namely, $\pi\circ\sigma(i)=\sigma(\pi(i))$, for all $1\leq i\leq n$.
Under this operation, the set~$S_n$ is a noncommutative group, known as the symmetric group
of order $n!$.
We denote by $\varepsilon\deff[1,2,\ldots,n]$ the identity permutation of~$S_n$.
Given a permutation $\sigma\in S_n$, an \emph{adjacent transposition}, $(i,i+1)$, for some $1\leq i\leq n-1$, is an exchange of the two
adjacent elements $\sigma(i)$ and $\sigma(i+1)$ in $\sigma$.
The result is the permutation $\pi=[\sigma(1),\ldots,{\sigma(i-1)},{\sigma(i+1)},\sigma(i),{\sigma(i+2)},\ldots,\sigma(n)]$.
Observe that the notation $(i,i+1)$ is also used for the cycle decomposition of the permutation $[1,2,\ldots,i-1,i+1,i,i+2,\ldots,n]$
and the permutation $\pi$ can also be written as ${\pi=(i,i+1)\circ \sigma}$.
In other words, left multiplication by $(i,i+1)$ exchanges the elements in positions $i,i+1$.
Right multiplication by $(i,i+1)$ exchanges the elements $i,i+1$.
Two adjacent transpositions $(i,i+1)$ and $(j,j+1)$ are called
\emph{disjoint} if either $i+1<j$ or $j+1<i$. For two permutations $\sigma,\pi \in S_n$, the Kendall $\tau$-distance between $\sigma$ and
$\pi$, $d_K(\sigma,\pi)$, is defined as the minimum number of adjacent transpositions needed
to transform $\sigma$ into $\pi$ \cite{KeGi90}.
For $\sigma\in S_n$, the Kendall $\tau$-weight of $\sigma$, $w_K(\sigma)$, is
defined as the Kendall
$\tau$-distance between $\sigma$ and the identity permutation $\varepsilon$.
The following expression for $d_K(\sigma,\pi)$ is well known~\cite{JSB10},~\cite{Knu98}.
%\vspace{-0.25cm}
%\begin{small}
\begin{equation}
\label{eq:distFormula}
d_K(\sigma,\pi)= \\ |\{(i,j)~:~\sigma^{-1}(i)<\sigma^{-1}(j)\wedge\pi^{-1}(i)>\pi^{-1}(j)\}|.
\end{equation}
%\end{small}

For a permutation $\sigma=[\sigma(1),\sigma(2),\ldots,\sigma(n)]\in S_n$, the \emph{reverse} of $\sigma$
is the permutation $\sigma^r\deff[\sigma(n),\sigma(n-1),\ldots,\sigma(2),\sigma(1)]$.
It follows from equation (\ref{eq:distFormula}) that for every $\sigma,\pi\in S_n$, $d_K(\sigma,\pi)\leq{n\choose 2}$ and $d_K(\sigma,\pi)={n\choose 2}$ if and only if $\pi=\sigma^r$.
The following lemma is an immediate consequence from the expression to compute the Kendall $\tau$-distance given in~(\ref{eq:distFormula}).
\begin{lemma}
\label{lem:reverse}
For every $\sigma,\pi\in S_n$,
$$
d_K(\sigma,\pi)+d_K(\sigma^r,\pi)=d_K(\sigma,\sigma^r)={n\choose 2}.
$$
\end{lemma}
%\vspace{-0.4cm}
The Kendall $\tau$-metric is right invariant \cite{Cay1878,DeHu98}, i.e.
for every three permutations $\sigma,\pi,\rho\in S_n$ we have
$d_K(\sigma,\pi)=d_K(\sigma\circ\rho,\pi\circ\rho)$. Note, that the Kendall $\tau$-metric is not left invariant.
The Kendall $\tau$-metric on~$S_n$ is graphic, i.e. for every two permutations $\sigma,\pi\in S_n$
their Kendall $\tau$-distance is equal to the length of the shortest path between $\sigma$ and $\pi$ in the graph
$G_n$, whose vertex set is the set $S_n$, and two vertices are connected by an edge if and only if their Kendall $\tau$-distance is one.

A distance measure $d(\cdot,\cdot)$ over a space $\cV$, is called \emph{bipartite} if every three elements $x,y,z\in \cV$
satisfy the equality $d(x,y)+d(y,z)\equiv d(x,z)~(\hbox{mod }2)$, i.e. the related graph is bipartite.
The Kendall $\tau$-metric on $S_n$ is bipartite as stated in the next lemma.

\begin{lemma}
\label{lem:biGraph}
The Kendall $\tau$-metric over $S_n$ is bipartite.
\end{lemma}

\begin{proof}
Just note that by~(\ref{eq:distFormula}) two permutations which differ in exactly one adjacent transposition
have different weights modulo 2. This implies that the related graph $G_n$ and the Kendall $\tau$-metric are bipartite.
\end{proof}

\begin{cor}
If $\sigma$ and $\pi$ are two permutations in $S_n$ then
$w_K(\sigma)+w_K(\pi)\equiv w_K(\sigma\circ\pi)~(\hbox{mod}~2)$.
\end{cor}

\begin{proof}
Since the Kendall $\tau$-metric is right invariant, it follows that
$w_K(\pi)=d_K(\pi,\epsilon)=d_K(\epsilon,\pi^{-1})=w_K(\pi^{-1})$.
Hence, by the definition of the Kendall $\tau$-weight and by Lemma~\ref{lem:biGraph}, we have that
\begin{equation}
\label{eq:rel1}
w_K(\sigma)+w_K(\pi)=w_K(\sigma)+w_K(\pi^{-1})=d_K(\sigma,\epsilon) + d_K(\pi^{-1},\epsilon) \equiv d_K(\sigma,\pi^{-1}) ~(\hbox{mod}~2)~.
\end{equation}
Since the Kendall $\tau$-metric is right invariant, it follows that
\begin{equation}
\label{eq:rel2}
d_K(\sigma,\pi^{-1})=d_K(\sigma\circ\pi ,\epsilon)=w_K(\sigma\circ\pi)
\end{equation}
Thus, by (\ref{eq:rel1}) and (\ref{eq:rel2}), we have that $w_K(\sigma)+w_K(\pi)\equiv w_K(\sigma\circ\pi)~(\hbox{mod}~2)$.
\end{proof}

Given a metric space, one can define codes. We say that $\cC\subseteq S_n$ has \emph{minimum distance} $d$
if $d_K(\sigma,\pi)\geq d$, for every two distinct permutations $\sigma,\pi\in\cC$.
For a given space $\cV$ with a distance measure $d(\cdot,\cdot)$,
a subset $C$ of $\cV$ is a \emph{perfect
code} with \emph{radius}~$R$ if for every element $x \in \cV$
there exists exactly one codeword $c \in C$ such that $d(x,c) \leq
R$. For a point $x \in \cV$, the \emph{ball} of radius~$R$
centered at $x$, $B(x,R)$, is defined by $B(x,R) \deff \{ y
\in \cV ~:~ d(x,y) \leq R \}$. In the Kendall $\tau$-metric
the size of a ball does not depend on the center of the
ball. This is a consequence of the fact that the Kendall $\tau$-distance
is right invariant. It is readily verified that
\begin{theorem}
\label{thm:sphere_bound} Let $\cV$ be a space with a
distance measure $d(\cdot,\cdot)$. For a code $C\subseteq \cV$ with minimum distance
$2R+1$ and a ball $B$ with radius $R$ we have $| C | \cdot |B | \leq | \cV |$, where $|S|$
is the size of the set $S$.
\end{theorem}

\vspace{0.1cm}

Theorem~\ref{thm:sphere_bound} is known as the \emph{sphere packing
bound} (even so it is really a ball packing bound). In a code $C$ which attains this bound, i.e.
$| C | \cdot |B | = | \cV |$, the balls with radius $R$ around
the codewords of $C$ form a partition of $\cV$. Such a code
is a perfect code. A perfect code with radius $R$ is also called a
\emph{perfect $R$-error-correcting code}.

Perfect codes is one of the most fascinating topics in coding
theory. These codes were mainly considered for the Hamming scheme,
e.g.~\cite{EtVa94,Mol86,Phe83,Phe84a,Phe84}. They were
also considered for other schemes such as the Johnson scheme, e.g.~\cite{Etz96,EtSc,Roos}, the
Grassmann scheme \cite{Chi87,MaZh95}, and to a larger extent also in the Lee and the
Manhattan metrics, e.g.~\cite{Etz11,GoWe70,Hor09,Post}. Note, that
the minimum distance of a perfect code is always an odd integer.
A more general concept in which codes can have even minimum distances as well,
is a diameter perfect code \cite{AAK01}.
This concept is based on Delsarte's code-anticode bound \cite{Delsarte}
for distance regular graphs. Since the Kendall $\tau$-metric over
$S_n$ does not induce a distance regular graph,
Delsarte's theorem may not apply for this metric.
However, an alternative proof shows that such type of a
bound is also valid for the Kendall $\tau$-metric.

\section{The Nonexistence of Some Perfect Codes}
\label{sec:nonexist}

In this section we prove that there are no single-error-correcting codes in $S_n$, where $n$ is a prime greater than 4.
Similarly, we also show that there are no perfect single-error-correcting
codes in $S_n$, for $4\leq n\leq 10$.

For each $i$, $1\leq i\leq n$, let $T_{n,i}\deff\{\sigma~:~\sigma\in S_n,~\sigma(i)=1\}$, i.e. $\sigma\in S_n$ is an element of $T_{n,i}$
if 1 appears in the $i$th position of $\sigma$. Clearly, $|T_{n,i}|=(n-1)!$.

Assume that there exists a perfect single-error-correcting code $\cC\subset S_n$. For each $i$, $1\leq i\leq n$, let
$$
\cC_i\deff\cC\cap T_{n,i}~~~~\text{ and }~~~~ x_i\deff|\cC_i|.
$$

We say that a codeword $\sigma\in \cC$ \emph{covers} a permutation $\pi\in S_n$ if $d_K(\sigma,\pi)\leq 1$.
Since $\cC$ is a perfect single-error-correcting code, it follows that
each permutation in~$T_{n,1}$ must be at distance at most one from
exactly one codeword of $\cC$ and this codeword must belong to either $\cC_1$ or $\cC_2$.
Every codeword $\sigma\in \cC_1$ covers exactly~${n-1}$ permutations in $T_{n,1}$.
It covers itself and the $n-2$ permutations in~$T_{n,1}$
obtained from $\sigma$ by exactly one adjacent transposition $(i,i+1)$, $1 < i < n$.
Each codeword $\sigma\in\cC_2$ covers exactly one permutation $\pi \in T_{n,1}$, $\pi=(1,2)\circ\sigma$.
Therefore, we have that

\begin{equation}
\label{eq:1r1}
(n-1)x_1+x_2=(n-1)!~.
\end{equation}

Similarly, by considering how the permutations of~$T_{n,n}$ are covered by
the codewords of $\cC$, we have that

%\vspace{-0.2cm}
\begin{equation}
\label{eq:1rn}
x_{n-1}+(n-1)x_n=(n-1)!~.
\end{equation}

For each $i$, $2\leq i\leq n-1$, each permutation in~$T_{n,i}$ is covered by
exactly one codeword that belongs to
either $\cC_{i-1}$, $\cC_{i}$, or $\cC_{i+1}$. Each codeword
$\sigma\in \cC_i$ covers exactly~${n-2}$ permutations in $T_{n,i}$.
It covers itself and the $n-3$ permutations in $T_{n,i}$
obtained from $\sigma$ by exactly one adjacent transposition $(j,j+1)$, where $1 \leq j<i-1$ or $i<j<n$.
Each codeword in $\cC_{i-1}\cup \cC_{i+1}$ covers exactly one permutation from $T_{n,i}$.
Therefore, for each $i$, ${2\leq i\leq n-1}$, we have that

%\vspace{-0.2cm}
\begin{equation}
\label{eq:1regular}
x_{i-1}+(n-2)x_i+x_{i+1}=(n-1)!~.
\end{equation}

Let $\mathbf{x}=(x_1,x_2,\ldots,x_n)$ and let $\mathbf{1}$ denote the all-ones column vector.
Equations (\ref{eq:1r1}), (\ref{eq:1rn}), and (\ref{eq:1regular}) can be written in a matrix
form as

%\vspace{-0.2cm}
\begin{equation}
\label{eq:linSystem}
A\mathbf{x}^T=(n-1)!\cdot\mathbf{1},
\end{equation}
where $A=(a_{i,j})$ is an $n\times n$ matrix defined by
\begin{footnotesize}
\begin{equation*}
A= \left( \scalemath{0.85}{\begin{array}{ccccccccc}
n-1 & 1 & 0& 0 & \cdots & 0 & 0&\ldots & 0 \\
1& n-2 & 1 &0 & \cdots & 0 & 0& \ldots &0\\
0 & 1 & n-2 & 1 & \cdots & 0 & 0& \ldots &0 \\
\vdots& \vdots & \vdots & \vdots & \ddots & \vdots & \vdots &\vdots  &\vdots \\
\vdots& \vdots & \vdots & \vdots & \ddots & \vdots & \vdots &\vdots  &\vdots \\
0 &\ldots&0 & 0 & \cdots & 1 & n-2 & 1 &0\\
0 &\ldots&0& 0 & \cdots &0 & 1 & n-2 & 1 \\
0 &\ldots&0 & 0 & \cdots  & 0 & 0& 1 & n-1
\end{array}}\right).
\end{equation*}
\end{footnotesize}

Since the sum of every row in $A$ is equal to $n$ it follows that the linear equation system defined in
(\ref{eq:linSystem}) has a solution ${\mathbf{y}^T=\frac{(n-1)!}{n}
\cdot\mathbf{1}}$.
We will show that if $n > 3$ then $A$ is a nonsingular matrix
and hence $\mathbf{y}$ is the unique solution of~(\ref{eq:linSystem}), i.e.
$\mathbf{x}=\mathbf{y}$.
To this end, we need the following theorem known as the Levy-Desplanques Theorem~\cite[p. 125]{HoJo91}.

\begin{theorem}
\label{lem:Gresh}
Let $B=(b_{i,j})$ be an $n\times n$ matrix. If $|b_{i,i}|>\sum_{j\neq i}|b_{i,j}|$
for all $i$, $1\leq i\leq n$, then $B$ is nonsingular.
\end{theorem}

For every $n>4$ we have that for each $i$, $1\leq i\leq n$, $a_{i,i}\geq n-2>2\geq \sum_{j\neq i} a_{i,j}$.
Hence, by Theorem~\ref{lem:Gresh} it follows that $A$
is nonsingular. For $n=4$ it can be readily verified that the matrix $A$ is nonsingular.
As a consequence we have that $\mathbf{x}^T=\frac{(n-1)!}{n} \cdot\mathbf{1}$ for every $n \geq 4$.
If $n=4$ or $n$ is a prime greater than $4$ then $\frac{(n-1)!}{n}$ is not an integer and therefore,
a perfect single-error-correcting code does not exist, i.e.

\begin{theorem}
\label{thm:nonexist}
There is no perfect single-error-correcting code in $S_n$, where $n>4$ is a prime or $n=4$.
\end{theorem}

\begin{remark}
It was brought to our attention that Theorem~\ref{thm:nonexist} is a special case of Theorem~5 in \cite{DeSe03}.
However, there is a crucial mistake in the proof of this theorem, which cannot be resolved.
The proof follows by induction on $n$, where the induction step is based on a partition of $S_n$ into ${n\choose k}$ classes, $2\leq k\leq n-2$, according to the set of the $k$ first elements in the permutations. It is stated that if $\cC\subset S_n$ is a code with minimum distance $3$ and $\cC$ is contained in one of these classes, then the projection of $\cC$ into $S_k$ has also minimum distance 3. This argument is clearly wrong. For example, the code $\{[1,2,3,4,5],[3,1,2,5,4]\}$ has minimum distance $3$ and the first three elements in each of its codewords belong to $\{1,2,3\}$. However, its projection into $S_3$ is the code $\{[1,2,3],[3,1,2]\}$, which has minimum distance 2. A similar example can be found for every $n\geq 4$ and for each $2\leq k\leq n-2$.
\end{remark}

The following theorem proved in~\cite{BuzThesis} implies that perfect
single-error-correcting codes must have a very symmetric and uniform structure. This might be useful to rule
out the existence of these codes for other parameters as well. The proof of this theorem is a generalization of the technique used to prove Theorem~\ref{thm:nonexist}. It is omitted here since the theorem is not used in the sequel.

\begin{theorem}
\label{thm:regularity}
Assume that there exists a perfect single-error-correcting code $\cC\subset S_n$,
where $n>11$. If $r<\frac{n}{4}$
then for each sequence of $r$ distinct elements of $[n]$,
$i_1,i_2,\ldots,i_r$, and for each set of $r$ positions, $1\leq j_1<j_2<\ldots<j_r\leq n$,
there are exactly $\frac{(n-r)!}{n}$ codewords $\sigma\in\cC$,
such that $\sigma(j_\ell)=i_\ell$, for each $\ell$, $1\leq \ell \leq r$.
\end{theorem}

For $n=6,8,9,10$, we use similar arguments and obtain systems
of linear equations. We used a computer to show that these systems
have no solutions over the nonnegative integers,
and to conclude that perfect single-error-correcting codes in $S_n$ do not exist for these values of $n$.
More details on these cases can be found in Appendix~A.

\begin{cor}\label{cor:sphere-1bound}
$P(n,3)<(n-1)!$ if $n$ is a prime greater than $4$ or $4\leq n\leq 10$.
\end{cor}

\begin{proof}
The size of a ball with radius one in $S_n$ , when the Kendall $\tau$-metric is used,
is $n$. Hence, by Theorem~\ref{thm:sphere_bound} and the discussion which follows this
theorem we have that, a single-error-correcting code $\cC \subset S_n$ is perfect if and only if
$|\cC|=(n-1)!$. Since such codes do not
exist if $n$ is a prime greater than 4 or if $4\leq n\leq 10$, it follows
that $P(n,3)<(n-1)!$.
\end{proof}

%%%%%%%%%%%%%%%%%%%%%%%%%%%%%%%%%%%%%%%%%5
%%5%%%%%%%%%%%%%%%%%%%%%%%%%%%%%%%%%%%%%
%%%
%%%  Anticodes and Diameter Perfect Codes
%%%%%%%%%%%%%%%%%%%%%%%%%%%%%
%%%%%%%%%%%%%%%%%%%%%%%%%%%%%55

\section{Anticodes and Diameter Perfect Codes}
\label{sec:diameter}

In all the perfect codes of a graphic metric the minimum distance of the code is an
odd integer. If the minimum distance of the code~$C$ is an even
integer then~$C$ cannot be a perfect code. The reason is
that for any two
codewords ${c_1 , c_2 \in C}$ such that $d(c_1,c_2)=2 \delta$, there
exists a word~$x$ such that $d(x , c_1)=\delta$ and
$d(x,c_2)=\delta$. For this case another concept is used, a
diameter perfect code, as was defined in~\cite{AAK01}. This
concept is based on the code-anticode bound presented by
Delsarte~\cite{Delsarte}. An \emph{anticode}~$\cA$ of \emph{diameter}~$D$
in a space~$\cV$ is a subset of words from~$\cV$ such that $d(x,y)
\leq D$ for all $x,y \in \cA$.

\begin{theorem}
\label{thm:abound} If a code $C$, in a space $\cV$ of a distance
regular graph, has minimum distance $d$ and in an anticode~$\cA$
of the space $\cV$ the maximum distance is $d-1$ then $|C |
\cdot |\cA| \leq | \cV |$.
\end{theorem}

\vspace{0.2cm}

Theorem~\ref{thm:abound} which was proved in~\cite{Delsarte} is
a generalization of Theorem~\ref{thm:sphere_bound} (the sphere packing bound)
and it can be applied to
the Hamming scheme since the related graph is distance regular (see \cite{BrCo89} for the definition of a distance regular graph).
It cannot be applied to the Kendall $\tau$-metric since the
related graph is not distance regular if $n > 3$. This can be easily verified
by considering the three permutations $\varepsilon=[1,2,3,4,5,\ldots,n]$, $\sigma=[3,1,2,4,5,\ldots,n]$,
and $\pi=[2,1,4,3,5,\ldots,n]$ in~$S_n$. Clearly,
$d_K(\varepsilon,\sigma)=d_K(\varepsilon,\pi)=2$ and
there exists exactly one permutation~$\alpha$ for which
$d_K (\varepsilon,\alpha)=1$ and $d_K (\alpha,\sigma)=1$, while there exist exactly two permutations
$\beta,\gamma$ for which $d_K (\varepsilon,\beta)=1$, $d_K (\beta,\pi)=1$,
$d_K (\varepsilon,\gamma)=1$, and $d_K (\gamma,\pi)=1$.
Fortunately, an alternative proof which was given in~\cite{AAK01}
and was modified in~\cite{Etz11}
will work for the Kendall $\tau$-metric.

\begin{theorem}
\label{thm:Ahlswede}
Let $\cC_{\cD}$ be a code in $S_n$
with Kendall $\tau$-distances between codewords taken from a set
$\cD$. Let $\cA \subset S_n$ and let $\cC'_{\cD}$ be the largest
code in $\cA$ with Kendall $\tau$-distances between codewords taken from the
set $\cD$. Then
$$
\frac{|\cC_{\cD}|}{n!} \leq \frac{|\cC'_{\cD}|}{|\cA|} ~.
$$
\end{theorem}

\begin{proof}
Let $\mathcal{B} \deff \{ (\sigma, \pi) ~:~ \sigma \in \cC_{\cD} , ~ \pi \in S_n,~ \sigma \circ \pi
\in \cA \}$. For a given codeword $\sigma\in \cC_{\cD}$ and a word
$\alpha \in \cA$, there is exactly one element $\pi\in
S_n$ such that $\alpha =\sigma\circ \pi$. Therefore, $| \mathcal{B} | = |\cC_{\cD}|
\cdot | \cA |$.

Since the Kendall $\tau$-metric is right invariant it follows that for every $\pi\in S_n$,
the set $\cC_{\pi}\deff\{\sigma\circ \pi~:~\sigma\in \cC_{\cD}\}$ has the same Kendall $\tau$-distances as in $\cC_{\cD}$, i.e.
the Kendall $\tau$-distances between codewords of $\cC_{\pi}$ are taken
from the set $\cD$.
Together with the fact that $\cC'_{\cD}$ is the largest code in $\cA$, with Kendall $\tau$-distances
between codewords taken from the set~$\cD$, it follows that for
any given word $\pi \in S_n$ the set $\{ \sigma ~:~ \sigma \in \cC_{\cD} ,
~\sigma\circ \pi \in \cA  \}$ has at most~$|\cC'_{\cD}|$ codewords. Hence, $|\mathcal{B}|
\leq |\cC'_{\cD}| \cdot n!$.

Thus, since $|\cB|=|\cC_{\cD}|\cdot |\cA|$, we have that $|\cC_{\cD}| \cdot | \cA | \leq |\cC'_{\cD}| \cdot n!$ and
the claim is proved.
\end{proof}

\begin{cor}
\label{cor:anti_Kendall} If a code $\cC \subseteq S_n$ has minimum Kendall
$\tau$-distance~$d$ and in an anticode $\cA \subset S_n$ the maximum
Kendall $\tau$-distance is~$d-1$ then $| \cC | \cdot |\cA| \leq n!$.
\end{cor}

\begin{proof}
Let $\cD = \{ d,d+1 , \ldots , {n\choose 2} \}$
and let $\cC_{\cD}\subseteq S_n$ be a code with minimum Kendall
$\tau$-distance $d$. Let $\cA$ be a subset of $S_n$ with Kendall $\tau$-distances
between words of $\cA$ taken from the set $\{ 1,2, \ldots , d-1
\}$, i.e. $\cA$ is an anticode with diameter $d-1$. Clearly, the
largest code in $\cA$ with Kendall $\tau$-distances from $\cD$ has only one
codeword. Applying Theorem~\ref{thm:Ahlswede} on $\cD$,
$\cC_{\cD}$, and $\cA$, implies that $|\cC_{\cD}| \cdot | \cA | \leq n!$.
\end{proof}

If there exists a code $\cC\subseteq S_n$ with minimum Kendall $\tau$-distance $d=D+1$
and an anticode $\cA$ with diameter~$D$ such that $|\cC|\cdot|\cA|=n!$ then
$\cC$ is called a \emph{$D$-diameter perfect} code. In this case, $\cA$ must be an
anticode with maximum distance (diameter)~$D$ of the largest possible size, and $\cA$ is called an
\emph{optimal} anticode of diameter~$D$. If $D=2R$ and the ball of radius $R$ is an optimal
anticode then a $D$-diameter perfect code is a perfect
$R$-error-correcting code. It is interesting to find
the optimal anticodes in $S_n$ and to determine their sizes. Using the sizes of such optimal anticodes we can
obtain by Corollary \ref{cor:anti_Kendall} upper bounds on $P(n,2\delta)$.
In the rest of this section we will mostly consider bounds on the size of optimal anticodes and use these bounds to obtain new upper bounds on $P(n,2\delta)$.
The proof of the next theorem is given in~Appendix~B.

\begin{theorem}\label{thm:diameter2Anticode}
Every optimal anticode with diameter 2 (using the Kendall $\tau$-distance) in~$S_n$, $n\geq 5$,
is a ball with radius one whose size is $n$.
\end{theorem}

We will now consider lower bounds on the size of optimal anticodes with odd diameter. These bounds will imply new lower bounds on $P(n,2\delta)$.
To this end we will define a double ball of radius $R$.
For a given space $\cV$ with a distance measure $d(\cdot,\cdot)$ and for two elements $x,y\in \cV$ such that $d(x,y)=1$, the \emph{double ball} of radius $R$ centered at $x$ and $y$ is defined by $DB(x,y,R)\deff B(\mathbf{x},R)\cup B(\mathbf{y},R)$. Let $B_{n,R}$ be a ball of radius $R$ in $S_n$. W.l.o.g., we may assume that $B_{n,R}=B(\varepsilon,R)$. For every $n\geq 1$ and $R\geq 0$, we denote by $DB_{n,R}$ the double ball of radius $R$ in $S_n$ centered at the identity permutation $\varepsilon$ and the permutation $(1,2)$.

\begin{lemma}
\label{lem:oddAnticode}
Let $\cV$ be a space with a distance measure $d(\cdot,\cdot)$.
For every $x,y\in \cV$ such that $d(x,y)=1$ we have
\begin{itemize}
\item[(1)] $DB(x,y,R)$ is an anticode of diameter at most $2R+1$.

\item[(2)] $|DB(x,y,R)|=|B(x,R)|+|B(y,R)|-|B(x,R)\cap B(y,R)|$.

\item[(3)] If $d(\cdot,\cdot)$ over $\cV$ is bipartite then
$B(x,R)\cap B(y,R)=DB(x,y,R-1)$.

\end{itemize}
\end{lemma}
\begin{proof}
$(1)$ follows immediately from the triangle inequality and $(2)$ is trivial.

If $z\in B(x,R)\cap B(y,R)$ then $d(x,z)\leq R$ and $d(y,z)\leq R$.
Assume that $d(\cdot,\cdot)$ is bipartite, i.e. every three elements $\hat{x},\hat{y},\hat{z}\in \cV$ satisfies the equation $d(\hat{x},\hat{y})+d(\hat{y},\hat{z})\equiv d(\hat{x},\hat{z})~(\hbox{mod }2)$.
If $d(x,z)=d(y,z)=R$ then $d(x,y)+d(y,z)\not \equiv d(x,z)~(\hbox{mod }2)$, a contradiction.
Hence, $d(x,z)\leq R-1$ or $d(y,z)\leq R-1$ and therefore, $z\in DB(x,y,R-1)$.

On the other hand, if $z\in DB(x,y,R-1)$ then $d(x,z)\leq R-1$ or $d(y,z)\leq R-1$ and since $d(x,y)=1$ it follows from the triangle inequality that $d(x,z)\leq R$ and $d(y,z)\leq R$. Therefore, $z\in B(x,R)\cap B(y,R)$.

Thus, $z \in B(x,R)\cap B(y,R)$ if and only if $z \in DB(x,y,R-1)$, i.e. $B(x,R)\cap B(y,R)=DN(x,y,R-1)$.
\end{proof}

\begin{cor}\label{cor:sizeDS}
$|DB_{n,R}|=2|B_{n,R}|-|DB_{n,R-1}|$.
\end{cor}
\begin{proof}
By Lemma~\ref{lem:oddAnticode} (2) we have
$|DB_{n,R}|=2|B_{n,R}|-|B(\varepsilon,R)\cap B((1,2),R)|$. By Lemma~\ref{lem:oddAnticode} (3) we have that $|B(\varepsilon,R)\cap B((1,2),R)|=DB_{n-1,R}$. Thus,
$|DB_{n,R}|=2|B_{n,R}|-|DB_{n,R-1}|$.
\end{proof}

\begin{theorem}
\label{thm:diameter3Anticode}
If $n\geq 4$ then
$ DB_{n,1}$
is an optimal anticode of diameter 3, whose size is ${2(n-1)}$.
\end{theorem}
\begin{proof}
The claim can be easily verified for $n=4$.
By the first part of Lemma \ref{lem:oddAnticode} and by Corollary~\ref{cor:sizeDS}
it follows that $DB_{n,1}$ is an anticode of diameter 3
and size $2(n-1)$.

Let $\cA$ be an optimal anticode of diameter 3 in $S_n$, where $n\geq 5$, and let
$$
\cA_e=\{\sigma\in \cA~:~w_K(\sigma)\equiv 0~(\hbox{mod }2)\},~~~
\cA_o=\{\sigma\in \cA~:~w_K(\sigma)\equiv 1~(\hbox{mod }2)\}.
$$
Since the Kendall $\tau$-metric is bipartite, it follows that $\cA_e$ and $\cA_o$ are anticodes of diameter 2.
If $n\geq 5$ then by Theorem \ref{thm:diameter2Anticode} it follows that $|\cA_e|\leq n$ ($|\cA_o|\leq n$, respectively) and
$|\cA_e|=n$ ($|\cA_0|=n$, respectively) if and only if $\cA_e$ ($\cA_0$, respectively) is a ball of radius one.
The anticodes $\cA_e$ and $\cA_o$ cannot be balls of radius one and therefore,
$|\cA_e|\leq n-1$ and $|\cA_o|\leq n-1$.
Thus, $|\cA|=|\cA_e|+|\cA_o|\leq 2(n-1)$, for $n\geq 5$.
\end{proof}

As a consequence of
Corollary \ref{cor:anti_Kendall} and the fact that $DB_{n,R}$ is an anticode of diameter $2R+1$ we have the following upper bound on $P(n,2\delta)$, which generally considerably improves the known upper bounds.

\begin{cor}\label{cor:boundGeneral}
$$
P(n,2(R+1))\leq \frac{n!}{|DB_{n,R}|}.
$$
\end{cor}

\begin{cor}\label{cor:bound4}
$$
P(n,4)\leq \frac{n!}{2(n-1)}.
$$
\end{cor}
Note, that $P(n,4)\geq \frac{(n)!}{2(2n-1)}$~\cite{JSB10} and hence the size of the best known code is within a factor of two from the new upper bound.

Note also,  that since we proved that $DB_{n,1}$ is an optimal anticode of diameter $3$, the upper bound of Corollary~\ref{cor:bound4} is the best bound that can be derived from Corollary~\ref{cor:anti_Kendall}.
An intriguing question is whether $B_{n,R}$ is an optimal anticode of diameter $D=2R$, where $0\leq R <\frac{{n\choose 2}}{2}$ and whether $DB_{n,R}$ is an optimal anticode of diameter $2R+1$, where $0\leq R< \frac{{n\choose 2}-1}{2}$. Table~\ref{tab:nD} present the sizes of the largest known anticodes of diameter $D$ in $S_n$, for $4\leq n\leq 12$ and $2\leq D\leq \max\left\{{n\choose 2},20\right\}$ . For even values of $D$, the bound is the size of the related ball of radius $\frac{D}{2}$ and was computed by computer. A formula to compute some of these values is given in~\cite{Knu98,Mui98} and also in~\cite{JSB10}. Odd values of $D$ were computed using Corollary~\ref{cor:sizeDS}. Related bounds on $P(n,d)$ will be presented in Section~\ref{sec:cyclic}.

\begin{table}[h!]
\centering
\resizebox{15cm}{!}{
\begin{tabular}{|c|c|c|c|c|c|c|c|c|c|c|c|c|c|c|c|c|c|c|c|}
\hline
\backslashbox{$n$}{$D$} & $2$ & $3$& $4$ & 5  & 6   & 7   & 8   & 9   & 10   & 11   & 12  & 13 & 14 & 15 & 16 & 17 & 18 & 19 & 20\\
\hline
                 4      & $4$ & 6 & 9  &  12 & 24  & -  & -  & -  &  -  & -   & -  & - &  -  &  - & -   & - &  -     & -  & -   \\
\hline
                 5      & $5$ & $8$ & 14  & 20 & 29 & 38 & 49 & 60 & 120 & - & - & - &  -  & -  & -   & - &  -    &  -  &  -  \\
\hline
                 6      & 6 & 10 & 20   & 30  & 49  &  68 &  98 &  128 & 169   &  210  &  259 & 308 &  360  & 720  &  -  & - & -     & -   & -   \\
\hline
                 7      & 7 & 12 & 27  & 42  & 76  & 110  & 174  & 238  & 343   &  448  & 602  & 756 &  961  & 1,166  &  1,416 &    1,666    & 1,947   &  2,228 & 2,520 \\
\hline
                 8      & 8 & 14 & 35  & 56  & 111  & 166  & 285  &  404 & 628   &  852  & 1,230  & 1,608 &   2,191 & 2,774  & 3,606   & 4,438 &   5,546    &   6,654 &  8,039  \\
\hline
                 9      & 9 & 16 &  44   & 72  &  155 &  238 &  440 & 642  & 1,068   & 1,494   & 2,298  & 3,102 &  4,489  & 5,876  &  8,095  & 10,314 &   13,640    &  16,966  &  21,671  \\
\hline
                10      & 10 & 18 & 54   & 90  & 209  &  328 &  649 &  970 &  1,717  & 2,464   & 4,015  & 5,566 &  8,504  & 11,442  & 16,599   & 21,756 &  30,239     & 38,722   &   51,909 \\
\hline
                11      & 11 & 20 & 65   & 110  & 274  & 438  &  923 & 1,408  & 2,640   & 3,872   & 6,655  & 9,438 &  15,159  & 20,880  & 31,758   & 42,636 &  61,997     &  81,358  &   113,906   \\
\hline
                12      & 12 & 22 &  77  & 132  & 351  & 570  & 1,274  & 1,978  &  3,914  &  5,850  & 10,569  & 15,288 &  25,728  & 36,168  & 57,486   & 78,804 &  119,483     &  160,162 & 233,389 \\
%\hline
%                 13      &   90  & 156  & 441 &  726 &  1,715 & 2,704 & 5,629   & 8,554   & 16,198 & 23,842 \\
%\hline
%                  14      &   104  & 182  & 542 & 908  & 2,260  & 3,612 & 7,889   & 12,166   &  &  \\
\hline
\end{tabular}}\caption{sizes of the largest known anticodes of diameter $D$ in $S_n$}
 \label{tab:nD}
\end{table}

For completeness, we will present in the next few results some simple optimal anticodes and the related perfect codes and diameter perfect codes in~$S_n$, which might be considered as trivial. If $D={n\choose 2}$ then an optimal anticode of diameter $D$ in $S_n$ is~$S_n$ itself.
Hence, if $\frac{{n\choose2}}{2}\leq R< {n\choose 2}$ then an optimal anticode with diameter $2R\geq {n\choose 2}$ is~$S_n$. Since $|B_{n,R}|<n!$, for $\frac{{n\choose2}}{2}\leq R< {n\choose 2}$, it follows that $B_{n,R}$ is not an optimal anticode with diameter $2R$.
Similarly, if $\frac{{n\choose2}-1}{2}\leq R< {n\choose 2}-1$ then $|DB_{n,R}|<n!$ and hence, $DB_{n,R}$ is not an optimal anticode with diameter $2R+1$.

\begin{theorem}\label{thm:dnchoose2-1}
$\cA\subset S_n$ is an optimal anticode of diameter ${n\choose 2}-1$
if and only if $\cA$ contains either $\sigma$ or $\sigma^r$, for each $\sigma\in S_n$.
\end{theorem}
\begin{proof}
If $\cA$ is an optimal anticode of diameter ${n\choose2}-1$ then by Lemma~\ref{lem:reverse}, for every $\sigma\in S_n$,
$\cA$ cannot contain
both $\sigma$ and $\sigma^r$. On the other hand, if $\pi\neq \sigma^r$ then $d_K(\sigma,\pi)\leq {n\choose 2}-1$.
Thus, the theorem follows.
\end{proof}

\begin{cor}\label{cor:dnchoose2-1}
An optimal anticode $\cA\subset S_n$ of diameter ${n\choose 2}-1$ has size $\frac{n!}{2}$ and
can be chosen in $2^{\frac{n!}{2}}$ different ways.
\end{cor}

\begin{cor}\label{cor:Pnnchoose2}
\begin{itemize}
\item[]
\item[$\bullet$] For each $\sigma\in S_n$, the set $\{\sigma,\sigma^r\}$ is a $D$-diameter perfect code, $D={n\choose2}-1$.

\item[$\bullet$] If $2R+1={n\choose 2}$ then $\{\sigma,\sigma^r\}$ is a perfect $R$-error-correcting code.

\end{itemize}
\end{cor}

\begin{theorem}\label{thm:PdLarge}
If $ \frac{2}{3}{n\choose 2}<d\leq {n\choose 2}$ then $P\left(n,d\right)=2$.
\end{theorem}

\begin{proof} Any code of the form $\{\sigma,\sigma^r\}$ has minimum Kendall $\tau$-distance at least $d$, and therefore $P(n,d)\geq 2$.

Assume to the contrary that $P(n,d)\geq 3$, i.e. there exists a code $\cC\subset S_n$ with minimum Kendall $\tau$-distance~$d$ and of size 3. Since the Kendall $\tau$-metric is right invariant, we can assume w.l.o.g. that $\cC=\{\varepsilon,\sigma,\pi\}$. We have that $d\leq w_K(\sigma)$ and $d\leq w_K(\pi)$ and $d\leq d_K(\sigma,\pi)$. By Lemma~\ref{lem:reverse} we have that $d_K(\sigma,\varepsilon^r)\leq {n\choose 2}-d $ and $d_K(\pi,\varepsilon ^r)\leq {n\choose 2}-d$. By the triangle inequality it follows that $d_K(\sigma,\pi)\leq 2{n\choose 2}-2d<2{n\choose 2}-2\frac{2}{3}{n\choose 2}<d$.

\end{proof}
\begin{cor}
If $2R={n\choose 2}-1$ then $B_{n,R}$ is an optimal anticode of diameter ${n\choose 2}-1$.
\end{cor}
\begin{proof}
Follows from Lemma~\ref{lem:reverse}, Theorem~\ref{thm:dnchoose2-1}, and Corollary~\ref{cor:dnchoose2-1}.
\end{proof}

\begin{lemma}
If $2R+1={n\choose 2}-1$ then $DB_{n,R}$ is an optimal anticode of diameter ${n\choose 2}-1$.
\end{lemma}
\begin{proof}
Recall that $\varepsilon$ and $(1,2)$ are the centers of $DB_{n,R}$.
By Theorem~\ref{thm:dnchoose2-1} it is sufficient to show that for every $\sigma\in S_n$, either $\sigma\in DB_{n,R}$ or $\sigma^r\in DB_{n,R}$.
If $w_{K}(\sigma)\leq R$ then by Lemma~\ref{lem:reverse} $w_K(\sigma^r)={n\choose 2}-w_K(\sigma)>R+1$ and therefore, $\sigma\in DB_{n,R}$ and $\sigma^r\not\in DB_{n,R}$. Similarly, if $w_{K}(\sigma)>R+1$ then $\sigma\not\in DB_{n,R}$ and $\sigma^r\in DB_{n,R}$. If $w_{K}(\sigma)=R+1$ then by Lemma~\ref{lem:reverse} $w_K(\sigma^r)=R+1$. By Lemma~\ref{lem:biGraph} and since $w_K((1,2))=1$ it follows that either $d_K(\sigma,(1,2))=R$ or $d_K(\sigma,(1,2))=R+2$. Similarly, either $d_K(\sigma^r,(1,2))= R$ or $d_K(\sigma^r,(1,2))= R+2$. By Lemma~\ref{lem:reverse} we conclude that either $d_K(\sigma,(1,2))= R$ or $d_K(\sigma^r,(1,2))= R$.
\end{proof}

The next theorem can be easily verified.

\begin{theorem}
Any set $\{\sigma,\pi\}$ such that $d_K(\sigma,\pi)=1$ is an optimal anticode of diameter one.
The set of all permutations of even Kendall $\tau$-weight, known as the alternating group, $A_n$, is a $1$-diameter perfect code.
Similarly, the set of all permutations of odd Kendall $\tau$-weight, $S_n\setminus A_n$, is an $1$-diameter perfect code.
These codes are the only $1$-diameter perfect codes in $S_n$.
\end{theorem}

%%%%%%%%%%%%%%%%%%

%%%%%%%%%%%%%%%%%%%%%%%%%%%%%%%%%%%%%%%%%5
%%5%%%%%%%%%%%%%%%%%%%%%%%%%%%%%%%%%%%%%
%%%
%%%  THE CYCLIC KENDALL'S TAU DISTANCE
%%%%%%%%%%%%%%%%%%%%%%%%%%%%%
%%%%%%%%%%%%%%%%%%%%%%%%%%%%%55

\section{Constructions of Large Codes and a Table of the Bounds}
\label{sec:cyclic}

In this section we present two large codes with minimum Kendall $\tau$-distance 3 in $S_5$ and $S_7$.
These two codes have large automorphism groups and can be represented only by one or two codewords, respectively.
We hope that the method in which we constructed these codes can be applied for other values of $n$ and minimum Kendall $\tau$-distance.
In addition, we present a table of the lower and upper bounds on $P(n,d)$ for small values of $n$. Throughout this section the positions and elements of permutations of length $n$ are taken from the set $\{0,1,2,\ldots,n-1\}$ (instead of the set $[n]$).

By Theorem~\ref{thm:nonexist}, there is no perfect single-error-correcting code in $S_5$,
using the Kendall $\tau$-distance.
However, if we add to the set of adjacent transpositions, which defines the Kendall $\tau$-metric, the transposition $(0,n-1)$, we obtain a new metric in which the code $\cC_5$, consists of the following 20 codewords, is a perfect single-error-correcting code in $S_5$.

\begin{center}
$[0,1,2,3,4],~~ [0,2,4,1,3],~~ [0,3,1,4,2],~~ [0,4,3,2,1]$\\
$[1,2,3,4,0],~~ [2,4,1,3,0],~~ [3,1,4,2,0],~~ [4,3,2,1,0]$\\
$[2,3,4,0,1],~~ [4,1,3,0,2],~~ [1,4,2,0,3],~~ [3,2,1,0,4]$\\
$[3,4,0,1,2],~~ [1,3,0,2,4],~~ [4,2,0,3,1],~~ [2,1,0,4,3]$\\
$[4,0,1,2,3],~~ [3,0,2,4,1],~~ [2,0,3,1,4],~~ [1,0,4,3,2]$\\
\end{center}

Note, that if $[\sigma(0),\sigma(1),\ldots,\sigma(4)]$ is a codeword then $[\sigma(1),\ldots,\sigma(4),\sigma(0)]$ and $[2\sigma(0),2\sigma(1),\ldots,2\sigma(4)]$ are also codewords, where the computations are performed modulo 5. Hence, this code can be represented by only one codeword $[0,1,2,3,4]$ and it has an automorphism group of size 20.
Note, also that the minimum Kendall $\tau$-distance of this code is at least 3 (since the Kendall $\tau$-distance can only be increased by removing the transposition (0,n-1)) and hence ,
\begin{theorem}\label{thm:P53}
$$
P(5,3)\geq 20.
$$
\end{theorem}

In general, we suggest to search for codes in $S_n$, for small $n$, $n$ prime, and small minimum Kendall $\tau$-distance as follows. We require that if $\sigma=[\sigma(0),\sigma(1),\ldots,\sigma(n-1)]$ is a codeword in the code $\cC$ then $[\sigma(1),\ldots,\sigma(n-1),\sigma(0)]$, $[\sigma(0)-1, \sigma(1)-1, \ldots, \sigma(n-1)-1]$, and $[\alpha\sigma(0),\alpha\sigma(1),\ldots, \alpha\sigma(n-1)]$ are also codewords, where the computations are done modulo $n$ and $\alpha$ is a primitive root modulo $n$. Note, that $[\sigma(0)-1, \sigma(1)-1, \ldots, \sigma(n-1)-1]=\sigma\circ [1,2,\ldots,n-1,0]$. A computer search for such a code is easier since the code has a large automorphism group.
We leave as a nice exercise to the reader to verify that a codeword in such a code represents either $n(n-1)$ codewords (if and only if $[0,1,\ldots,n-1]$ is one of the represented codewords, as in $\cC_5$) or $n^2(n-1)$ codewords.

\begin{theorem}\label{thm:P57}
$$
P(7,3)\geq 588.
$$
\end{theorem}

\begin{proof}
Verify that the two representatives $\mu=[0,1,3,2,5,6,4]$ and $\nu=[0,1,2,3,6,4,5]$ yield the require code of size 588.
\end{proof}

The previous known lower bounds on $P(5,3)$ and $P(7,3)$ were $18$ and $526$, respectively~\cite{JSB10}. We summarise with
the best known bounds on $P(n,d)$, for $5\leq n\leq 7$ and $3\leq d\leq 9$, which are presented in Table~\ref{tab:bounds}.

\begin{table}[h!]
\centering
\resizebox{15cm}{!}{
\begin{tabular}{|c|c|c|c|c|c|c|c|}
\hline
{\backslashbox{$n$}{$d$}} & 3           & 4                  & 5           & 6                   & 7          & 8          & 9          \\
\hline
                 5      & $^f20-23^b$   &  $^h10-15^c$     &$^d6 -8^a$     &$^j4-6^c$       &$^i2^i$      &$^i2^i$     &$^i2^i$            \\
\hline
                 6      & $^d90-119^b$    &  $^h45-72^c$       & $^d23-36^a$     & $^h12 -24^c$     & $^d10 -14^a$    &$^h5-10^c$    &$^d4 -7^a$        \\
\hline
                 7      & $^e588-719^b$    &  $^h294-420^c$     &  $^d110-186^a$   & $ ^h55-120^c $   & $^d34-66^a$    &$^h17-45^c$    & $^d14-28^a$          \\
\hline
\end{tabular}}
\begin{itemize}
\item[]
\item[$\bullet$] a - The sphere packing bound.
\item[$\bullet$] b - The sphere packing bound + Theorem~\ref{thm:nonexist}.
\item[$\bullet$] c - Corollary~\ref{cor:boundGeneral}.
\item[$\bullet$] d - Lower bounds from ~\cite{JSB10}.
\item[$\bullet$] f - Theorem~\ref{thm:P53}.
\item[$\bullet$] e - Theorem~\ref{thm:P57}.
\item[$\bullet$] h - $P(n,2\delta)\geq \frac{1}{2}P(n,2\delta-1)$~\cite{JSB10}.
\item[$\bullet$] i - Theorem~\ref{thm:PdLarge}.
\item[$\bullet$] j - $\cC=\{[1,2,3,4,5],[1,5,2,3,4],[2,3,4,1,5],[1,4,3,2,5]\}$.

\end{itemize}
\caption{Best known lower and upper bound on $P(n,d)$.}
\label{tab:bounds}
\end{table}

%%%%%%%%%%%%%%%%%%%%%%%%%%%%%%%%%%
%%%
%%%       Group Codes
%%%
%%%%%%%%%%%%%%%%%%%%%%%%%%%%%%%%%%%%

%%%%%%%%%%%%%%%%%%%%%%%%%%%%%%%%%%
%%%
%%%       CONCLUSION
%%%
%%%%%%%%%%%%%%%%%%%%%%%%%%%%%%%%%%%%

\section{Conclusions and Open Problems}
\label{sec:conclusion}
We have considered several questions related to bounds on the size of codes in
the Kendall $\tau$-metric. We gave a novel technique to exclude
the existence of perfect single-error-correcting codes using the Kendall $\tau$-metric.
We applied this technique to prove that there are no perfect
single-error-correcting codes in $S_n$, where $n >4$ is a prime or
$4 \leq n \leq 10$, using the Kendall $\tau$-metric.
We examine the existence question
of diameter perfect codes in $S_n$ and the sizes of optimal anticodes with the Kendall $\tau$-distance.
We obtained a new upper bound
on the size of a code in $S_n$ with even Kendall $\tau$-distance.
Finally, we constructed two large codes with large automorphism groups in $S_5$ and $S_7$.

Our discussion raises many open
problems from which we choose a few as follows.

\begin{enumerate}
\item Prove the nonexistence of perfect codes in $S_n$, using the
Kendall $\tau$-metric, for more values of $n$ and/or other distances.

\item Do there exist more $D$-diameter perfect codes in $S_n$ with the Kendall $\tau$-metric, for $2\leq D<{n\choose 2}-1$?
We conjecture that the answer is no.

\item Is a ball with radius $R$ in $S_n$ always optimal
as an anticode with diameter $2R$ in $S_n$, for $2\leq R <\frac{{n\choose 2}}{2}$?

\item Is the double ball with radius $R$ in $S_n$ always optimal as an anticode with
diameter $2R+1$ in $S_n$, for $2\leq R<\frac{{n\choose 2}-1}{2}$?

\item What is the size of an optimal anticode in $S_n$ with
diameter $D$?

\item Improve the lower bounds on the sizes of codes in $S_n$
with even minimum Kendall $\tau$-distance.

\item Can the codes in $S_5$ and $S_7$ from Section~\ref{sec:cyclic} be generalized for higher values of $n$ and to larger distances? Are these codes of optimal size?
\end{enumerate}

\section*{Acknowledgment}
Sarit Buzaglo would like to thank Amir
Yehudayoff for many useful discussions. The authors would
like to thank the anonymous reviewer of the 2014 International Symposium on Information Theory
for valuable comments. They thank Simon Litsyn for bringing valuable references to their attention. The authors also thank three anonymous reviewers whose detailed reviews and comments helped to improve the presentation of this paper.
Finally, the authors want to thank Professor Hal Sudborough who found an error in the permutations of Theorem~\ref{thm:P57}
in an earlier version.

\section*{Appendix A}%

In Theorem~\ref{thm:nonexist} we proved that
a perfect single-error-correcting code in $S_n$ with the Kendall
$\tau$-metric does not exist if $n>4$ is a prime or if $n=4$.
The proof of Theorem~\ref{thm:nonexist} is based
on a certain linear equations system,
where the existence of a perfect single-error-correcting code in $S_n$
implies the existence of a solution to the linear equations system over the integers,
and thus, by showing the nonexistence of such solution we derive the nonexistence of a perfect
single-error-correcting code.
By using similar techniques
we prove the nonexistence of perfect single-error-correcting codes in $S_n$
for $n\in\{6,8,9,10\}$. For each such $n$, let $\cC$ be a perfect single-error-correcting code in $S_n$.
We will describe the corresponding linear equations system
and use a computer to show that this linear equations system does not have
a solution over the integers.
\begin{itemize}
\item[{$n=6$:}] We denote by $D_6$ the set of all vectors of
$ \{1,2,3\}^6$ in which each of
the elements 1,2,3 appears twice. For each $\mathbf{v}\in D_6$ we define $
S_\mathbf{v}$ to be the set of eight permutations
in $S_6$, such that the elements $1$ and $2$ appear in the two positions in which $1$
appears in $\mathbf{v}$, the elements $3$ and $4$ appear in the two positions in which $2$ appears in
$\mathbf{v}$, and the elements $5$ and $6$ appear in the two positions in which $3$ appears in $\mathbf{v}$.
Let $x_{\mathbf{v}}=|\cC\cap S_{\mathbf{v}}|$ and let $\mathbf{x}=(x_{\mathbf{v}_1},x_{\mathbf{v}_2},\ldots,x_{\mathbf{v}_m})$, where $m=|D_6|=\frac{6!}{2!2!2!}$. By considering how the elements of $S_{\mathbf{v}}$ are covered (similarly to the way it was done in the proof of Theorem~\ref{thm:nonexist}), for each
$\mathbf{v}\in D_6$, we obtain a linear equations system
of the form $A\mathbf{x}^T=|S_{\mathbf{v}}|\cdot\mathbf{1}=8\cdot \mathbf{1}$, where $A$ is a square matrix of order $m$.
The kernel of $A$ is an one-dimensional
vector space which is spanned by a vector $\mathbf{y}\in \{0,-1,1\}^9$, that has both negative and positive entries.
Every solution for this system is of the form $\frac{8}{6}\cdot\mathbf{1}+\alpha\cdot \mathbf{y}$, $\alpha\in \mathbb{R}$, and therefore, the system does not have a solution in which all entries are integers.

\item[{$n=8$:}]  We denote by $D_8$ the set of all vectors $\mathbf{v}\in \{1,2,3,4\}^8$ in which each of
the elements 1 and 2 appears three times and each of the elements $3$ and $4$ appears once.
For every $\mathbf{v}\in D_8$ we define $
S_\mathbf{v}$ to be the set of 36 permutations
in $S_8$, such that the elements $1,2,$ and $3$ appear in the three positions in which $1$
appears in $\mathbf{v}$, the elements $4,5,$ and $6$ appear in the three positions in which $2$ appears in
$\mathbf{v}$, the element $7$ appears in the position of $3$ in $\mathbf{v}$, and the
element $8$ appears in the position of $4$ in $\mathbf{v}$. Let $x_{\mathbf{v}}=|\cC\cap S_{\mathbf{v}}|$
and let $\mathbf{x}=(x_{\mathbf{v}_1},x_{\mathbf{v}_2},\ldots,x_{\mathbf{v}_m})$, where $m=|D_8|=\frac{8!}{3!3!}$.
By considering how elements of $S_{\mathbf{v}}$ are covered, for each
$\mathbf{v}\in D_8$, we obtain a linear equations system
of the form $A\mathbf{x}^T=36\cdot\mathbf{1}$,
where $A$ is a square matrix of order~$m$. The system has a unique solution, $\mathbf{x}^T=\frac{36}{8}\cdot\mathbf{1}$, which has non-integer entries.

\item[{$n=9$:}] We denote by $D_9$ the set of all vectors $\mathbf{v}\in \{1,2,3\}^9$ in which
the element 1 appears five times and each of the elements $2$ and $3$ appears twice.
For every $\mathbf{v}\in D_9$ we define $
S_\mathbf{v}$ to be the set of 480 permutations
in $S_8$, such that the elements $1,2,3,4,$ and $5$ appear in the five positions in which $1$
appears in $\mathbf{v}$, the elements $6$ and $7$ appear in the two positions in which $2$ appears in
$\mathbf{v}$, and the elements $8$ and $9$ appear in the two positions in which $3$ appears in $\mathbf{v}$.
Let $x_{\mathbf{v}}=|\cC\cap S_{\mathbf{v}}|$
and let $\mathbf{x}=(x_{\mathbf{v}_1},x_{\mathbf{v}_2},\ldots,x_{\mathbf{v}_m})$, where $m=|D_9|=\frac{9!}{5!2!2!}$. By considering how elements of $S_{\mathbf{v}}$ are covered, for each
$\mathbf{v}\in D_9$, we obtain a linear equations system
of the form $A\mathbf{x}^T=480\cdot\mathbf{1}$,
where $A$ is a square matrix of order $m$. The system has a unique solution, $\mathbf{x}^T=\frac{480}{9}\cdot\mathbf{1}$, which has non-integer entries.

\item[{$n=10$:}] We denote by $D_{10}$ the set of all vectors $\mathbf{v}\in \{1,2,3\}^{10}$ in which each of
the elements 1 and 2 appears four times and the element $3$ appears twice.
For every $\mathbf{v}\in D_{10}$ we define $
S_\mathbf{v}$ to be the set of 1,152 permutations
in $S_{10}$, such that the elements $1,2,3,$ and $4$ appear in the four positions in which $1$
appears in $\mathbf{v}$, the elements $5,6,7,$ and $8$ appear in the four positions in which $2$ appears in
$\mathbf{v}$, and the elements $9$ and $10$ appear in the two positions in which $3$ appears in $\mathbf{v}$.
Let $x_{\mathbf{v}}=|\cC\cap S_{\mathbf{v}}|$ and let $\mathbf{x}=(x_{\mathbf{v}_1},x_{\mathbf{v}_2},\ldots,x_{\mathbf{v}_m})$, where $m=|D_{10}|=\frac{10!}{4!4!2!}$. By considering how elements of $S_{\mathbf{v}}$ are covered, for each
$\mathbf{v}\in D_{10}$, we obtain a linear equations system
of the form $A\mathbf{x}^T=1,152\cdot\mathbf{1}$,
where $A$ is a square matrix of order $m$.
The system has a unique solution, $\mathbf{x}^T=\frac{1,152}{10}\cdot\mathbf{1}$, which has non-integer entries.

\end{itemize}
% Put the path to your bib file (or whatever else you use) here

\section*{Appendix B}

The purpose of this appendix is to prove Theorem~\ref{thm:diameter2Anticode} given in Section~\ref{sec:diameter}.

\noindent{\textbf{Theorem 7.}}\emph{
Every optimal anticode with diameter 2 (using the Kendall $\tau$-distance) in~$S_n$, $n\geq 5$,
is a ball with radius one whose size is $n$.}

\begin{lemma}\label{lem:neighbors}
Let $\sigma=(i,i+1)\circ(i+1,i+2)$ and let $\rho\neq \sigma$ be a permutation of weight 2 and distance 2 from $\sigma$. Then $\rho=(j,j+1)\circ(i+1,i+2)$ or $\rho=(i+1,i+2)\circ(i,i+1)$.
\end{lemma}

\begin{proof}
Recall first that for any two permutations $\alpha,\beta$, $d_K(\alpha,\beta)=1$ if and only if there exists an adjacent transposition $(k,k+1)$, such that $\alpha=(k,k+1)\circ\beta$. We distinguish between four cases. In the first two cases the permutation $\rho$ is at distance 2 from $\sigma$.
\begin{itemize}
\item[I.] $\rho=(j,j+1)\circ(i+1,i+2)$. In this case $\sigma=(i,i+1)\circ(j,j+1)\circ\rho$ and therefore $d_K(\sigma,\rho)\leq2$. By Lemma~\ref{lem:biGraph} we have that the Kendall $\tau$-metric is bipartite and since $\sigma$ and $\rho$ are both of even weight it follows that $d_K(\sigma,\rho)\geq2$. Thus, $d_K(\sigma,\pi)=2$.
\item[II.] $\rho=(i+1,i+2)\circ(i,i+1)$. In this case we have that $\sigma=\rho\circ\rho$ and similarly it follows that $d_K(\sigma,\rho)=2$.
\item[III.]  If $\rho=(j,j+1)\circ(k,k+1)$, where $j\neq k$ and $j,k\neq i+1$, then by~(\ref{eq:distFormula}) we have that $d_K(\sigma,\rho)\geq |\{(i+2,i),(i+2,i+1), (k,k+1)\}|>2$.
\item[IV.] If $\rho=(i+1,i+2)\circ(j,j+1)$. We distinguish be between four subcases.
\begin{itemize}
\item[1)]  If $j\not\in\{ i,i+1,i+2\}$, then $\rho=(j,j+1)\circ(i+1,i+2)$ and this case was considered in I.
\item[2)]  $j=i$ was considered in II.
\item[3)]  If $j=i+1$ then $\rho=\varepsilon $, i.e $w_K(\rho)=0$.
\item[4)]  If $j=i+2$ then $\rho=(i+1,i+2)\circ(i+2,i+3)$ and by~(\ref{eq:distFormula}) we have that $d_K(\sigma,\rho)= |\{(i+2,i), (i+2,i+1),(i+1,i+3),(i+2,i+3)\}|=4$.
    \end{itemize}
    \end{itemize}
\end{proof}

\begin{lemma}\label{lem:ii+1i+1i+2}
Let $\sigma=(i,i+1)\circ(i+1,i+2)$ and $\pi=(i+1,i+2)\circ(i,i+1)$, where $i\in [n-2]$, and let $\rho$ be a permutation of weight 2, $\rho\neq \sigma$ and $\rho\neq \pi$. Then either $d_K(\sigma,\rho)\geq 4$ or $d_K(\pi,\rho)\geq 4$.
\end{lemma}

\begin{proof}
By Lemma~\ref{lem:neighbors} it follows that if $d_K(\sigma,\rho)=2$ then $\rho=(j,j+1)\circ(i+1,i+2)$ or $\rho=\pi$. By symmetry it follows that if $d_K(\pi,\rho)=2$ then $\rho=(j,j+1)\circ(i,i+1)$ or $\rho=\pi$. Hence, there is no permutation $\rho$ of weight 2 and distance 2 from both $\sigma$ and $\pi$. By Lemma~\ref{lem:biGraph} we also have that the Kendall $\tau$-metric is bipartite and we conclude that any permutation of weight 2 other then $\sigma$ and $\pi$ must be at distance at least four from $\sigma$ or $\pi$.
\end{proof}

\begin{lemma}\label{lem:weight2}
Let $\cA$ be an anticode in $S_n$ with diameter 2 such that $\varepsilon\in \cA$, and let $\cB$ be the set of all permutations of weight 2 in $\cA$.
If $|\cB|\geq 4$ then $\cB$ is contained in a ball of radius one centered at some permutation $\sigma \in S_n$ of weight one.
\end{lemma}
\begin{proof}
If there exists some $i\in [n-2]$ such that
$(i,i+1)\circ(i+1,i+2),(i+1,i+2)\circ(i,i+1)\in \cB$,
then by Lemma~\ref{lem:ii+1i+1i+2} any other permutation of weight 2 is at distance at least four from either
$(i,i+1)\circ(i+1,i+2)$ or $(i+1,i+2)\circ(i,i+1)$, and therefore $|\cB|= 2$.

If for some $i\in [n-2]$ either $(i,i+1)\circ(i+1,i+2)$ or $(i+1,i+2)\circ(i,i+1)$ belongs to $\cB$, say w.l.o.g.
$(i,i+1)\circ(i+1,i+2)\in\cB$, then every permutation of $\cB\setminus\{(i,i+1)\circ(i+1,i+2)\}$
must be at distance 2 from $(i,i+1)\circ(i+1,i+2)$, and by Lemma~\ref{lem:neighbors} it follows that every such permutation
must be of the form $(j,j+1)\circ (i+1,i+2)$ for some $j\not\in \{i,i+1\}$. Therefore, $\cB\subset B((i+1,i+2),1)$.

If each permutation of $\cB$ is a multiplication of two
disjoint adjacent transpositions then let $\rho=(i,i+1)\circ(j,j+1)\in \cB$, where $j\not\in \{i-1,i,i+1\}$.
Hence, all permutations of $\cB$ are of the form $(\ell,\ell+1)\circ(j,j+1)$, where $\ell\not\in\{j,j+1\}$,
or $(\ell,\ell+1)\circ(i,i+1)$, where $\ell\not\in\{i,i+1\}$.
Assume w.l.o.g. that $\pi=(\ell,\ell+1)\circ(j,j+1)\in \cB$, $\pi\neq \rho$.
If every permutation of $\cB$ is of the form $(k,k+1)\circ(j,j+1)$ then $\cB\subset B((j,j+1),1)$.
Otherwise, the only possible other permutation of $\cB$ is $(i,i+1)\circ (\ell, \ell+1)$ and hence $|\cB|\leq 3$.

Thus, if $|\cB|\geq 4$ then $\cB\subset B(\sigma,1)$, for some $\sigma$ of weight one.
\end{proof}

\emph{\textbf{Proof of Theorem~\ref{thm:diameter2Anticode}:}} Let $\cA\subset S_n$, $n\geq 5$, be an anticode of diameter 2.
The Kendall $\tau$-metric is right invariant and hence w.l.o.g. we can assume
that $\varepsilon\in \cA$. Therefore, all the permutations of $\cA$ are of weight at most two.
We distinguish between four cases:

\begin{itemize}
\item[Case 1:] If $\cA$ does not contain a permutation of weight one then by Lemma \ref{lem:weight2} it follows that
$\cA$ is contained in a ball
of radius one centered at a permutation of weight one or $|\cA|\leq 4$.

\item[Case 2:] If $\cA$ contains exactly one permutation $\sigma\in S_n$ of weight one then by Lemma \ref{lem:biGraph}, the distance between $\sigma$ and any permutation of weight 2 is an odd integer and therefore, all permutations of weight 2 in $\cA$ must be
at distance one from $\sigma$. Thus, $\cA\subseteq B(\sigma,1)$.

\item[Case 3:] If $\cA$ contains two permutations of weight one, $\sigma=(i,i+1)$ and $\pi=(j,j+1)$, where $\sigma$ and $\pi$ are disjoint transpositions, then the only permutation of weight 2 and distance one from both $\sigma$ and $\pi$ is $(i,i+1)\circ(j,j+1)$ and therefore $\cA$ cannot contain more than one permutation of weight 2, hence $|\cA|\leq 4$.

\item[Case 4:] If $\cA$ contains two permutations of weight one, $\sigma=(i,i+1)$ and $\pi=(i+1,i+2)$, for some $i\in[n-2]$, then there is no permutation of weight 2 and distance one from both $\sigma$ and $\pi$ and therefore $\cA$ cannot contain permutations of weight 2, hence $|\cA|\leq 3$.

\item[Case 5:] If $\cA$ contains at least three permutations of weight one then $\cA$ cannot contain permutations of weight 2 and therefore $\cA\subseteq B(\varepsilon,1)$.
\end{itemize}

Thus, we proved that either $\cA$ is contained in a ball of radius one or $|\cA|\leq 4$. Since the size of a ball of radius one in $S_n$ is $n$, it follows that  if $n\geq 5$ then every optimal anticode of diameter 2 in $S_n$ is a ball of radius one. $\hfill{\Box}$

%\bibliography{allbib,extra}

\end{document}